\documentclass[11p,reqno]{amsart}

\usepackage[T1]{fontenc}
\usepackage{graphicx}
\usepackage{amssymb,amsthm,amsmath,mathrsfs,bm,braket,marginnote}
\usepackage{enumerate}
\usepackage{appendix}
\usepackage{bbm}
\usepackage{eufrak}
\usepackage{mathrsfs}
\usepackage[colorlinks=true, pdfstartview=FitV, linkcolor=blue, citecolor=blue, urlcolor=blue]{hyperref}

\usepackage{pgf}
\usepackage{pgfplots}
\usepackage{extarrows}
\usepackage{tikz}
\usetikzlibrary{arrows,calc}
\usepackage{verbatim}
\usetikzlibrary{decorations.pathreplacing,decorations.pathmorphing}

\numberwithin{equation}{section}
\newtheorem{theorem}{Theorem}[section]

\newtheorem{proposition}[theorem]{Proposition}

\newtheorem{remark}[theorem]{Remark}
\newtheorem{corollary}[theorem]{Corollary}

\newcommand{\p}{\partial}

\newcommand{\al}{\alpha}
\newcommand{\be}{\beta}
\newcommand{\tq}{{\theta_1}}
\newcommand{\tw}{{\theta_2}}

\newcommand{\du}{d\mu}

 \reversemarginpar
\begin{document}

\title[Generalisations of Cauchy bi-orthogonal polynomials and integrable lattices]{Two-parameter generalisations of Cauchy bi-orthogonal polynomials and integrable lattices}
\author{Xiang-Ke Chang}
\address{LSEC, ICMSEC, Academy of Mathematics and Systems Science, Chinese Academy
of Sciences, PO Box 2719, Beijing 100190, People's Republic of China}
\address{School of Mathematical Sciences, University of Chinese Academy of Sciences,
Beijing 100049, People's Republic of China
}
\email{changxk@lsec.cc.ac.cn}
\author{Shi-Hao Li}
\address{}
\email{lishihao@lsec.cc.ac.cn}
\author{Satoshi Tsujimoto}
\address{Department of Applied Mathematics and Physics, Graduate School of Imformatics, Kyoto University, Yoshida-Honmachi, Kyoto, Japan 606-8501}
\email{tujimoto@i.kyoto-u.ac.jp}
\author{Guo-Fu Yu}
\address{School of Mathematical Sciences, Shanghai Jiaotong University, People's Republic of China.}
\email{gfyu@sjtu.edu.cn}

\date{}

\dedicatory{}


\begin{abstract}
In this article, we consider the generalised two-parameter Cauchy two-matrix model and corresponding integrable lattice equation. It is shown that with parameters chosen as $1/k_i$ when $k_i\in\mathbb{Z}_{>0}$ ($i=1,\,2$), the average characteristic polynomials admit $(k_1+k_2+2)$-term recurrence relations, which provide us spectral problems for integrable lattices. The tau function is then given by the partition function of the generalised Cauchy two-matrix model as well as Gram determinant. The simplest example with exact solvability is demonstrated.
\end{abstract}

\maketitle
\section{Introduction}
The Cauchy two-matrix model was proposed during the studies of integrable systems, specifically on the analysis in the peakon solution of the Degasperis-Procesi equation \cite{lundmark03}. Later on, this matrix model has attracted much attention from groups in random matrix theory and many properties have been systematically studied. For example, the limiting behaviour of the Cauchy two-matrix model and the corresponding Riemann-Hilbert problem of the Cauchy bi-orthogonal polynomials were done in \cite{bertola092,bertola10,bertola09}; its connections with Bures ensemble in the levels of partition functions and 
correlation functions were demonstrated in \cite{forrester16,forrester19}. The joint probability density function (jPDF) of this model has the form
\begin{align*}
\frac{\prod_{1\leq j<k\leq N}(x_k-x_j)^2(y_k-y_j)^2}{\prod_{j,k=1}^N(x_j+y_k)}\prod_{j,k=1}^N\omega_1(x_j)\omega_2(y_k),\quad x_j,\,y_k\in\mathbb{R}_+
\end{align*}
with some non-negative weight functions $\omega_1$ and $\omega_2$. As is known, the development of random matrix model gives insights into orthogonal polynomials theory and classical integrable systems, e.g. the considerations of Hermitian matrix model with unitary invariance are closely related to orthogonal polynomials, and KP/1d-Toda hierarchy \cite{adler97,gerasimov91,tsujimoto00}, and it is interesting to know whether there are also some integrable systems behind the Cauchy two-matrix models. Very recently, the answer was given in \cite{li19}. It was shown that if $\omega_1=\omega_2$ and some proper time flows are involved, then the time-dependent partition function of the Cauchy two-matrix model can be regarded as the tau function of the CKP hierarchy as well as the so-called Toda lattice of the CKP type (C-Toda lattice for brevity). Moreover, the average characteristic polynomials of the Cauchy two-matrix model---the Cauchy bi-orthogonal polynomials $\{P_n(x)\}_{n=0}^\infty$ can act as the wave functions, and the remarkable four-term recurrence relation (with $a_n$, $b_n$, $c_n$ and $d_n$ properly chosen)
\begin{align}\label{1}
x(P_{n+1}(x)+a_nP_n(x))=P_{n+2}(x)+b_nP_{n+1}(x)+c_nP_n(x)+d_nP_{n-1}(x)
\end{align}
provides a $3\times 3$ spectral problem for the integrable system.

Compared with the previous work, there are two main motivations for our present studies. One is from the recent work on the integrable system related to a new class of extended affine Weyl group $\tilde{W}^{(k,k+1)}(A_l)$ \cite{zuo19}. As is known, the extended affine Weyl group $\tilde{W}^{(k)}(A_l)$ is related to the spectral operator (where $\Lambda$ is the shift operator)
\begin{align*}
\Lambda^k+a_{1}\Lambda^{k-1}+\cdots+a_{l+1}\Lambda^{k-l-1}, \quad 1\leq k<l,\, a_{l+1}\not=0,
\end{align*}
which can be regarded as the spectral problem of  the bigraded Toda hierarchy \cite{carlet06} and it has intimate connection with Muttalib-Borodin model, which can be viewed as a $\theta$-deformed model of Hermitian matrix model with unitary invariance. Regarding the extended affine Weyl group $\tilde{W}^{(k,k+1)}(A_l)$, the spectral problem is given by \cite{zuo19}
\begin{align}\label{sspectral}
(\Lambda-a_{l+2})^{-1}(\Lambda^{k+1}+a_1\Lambda^k+\cdots+a_{l+1}\Lambda^{k-l}),\quad 1\leq k<l,\, a_{l+1}a_{l+2}\not=0.
\end{align}
Obviously, the four-term recurrence relation \eqref{1} is a very special example with $k=1$ and $l=2$. It inspires us to consider whether there are any orthogonal polynomials system admitting such general recurrence relations and to find corresponding integrable lattices related to the general spectral problem.

Another motivation is from a very recent work on the $\theta$-deformation of the Cauchy two-matrix model \cite{forrester19}. The idea to generalise the Cauchy two-matrix model is to consider the jPDF
\begin{align}\label{cauchy}
\frac{\prod_{1\leq j<k\leq N}(x_k-x_j)(y_k-y_j)}{\prod_{j,k=1}^N(x_j+y_k)}\prod_{1\leq j<k\leq N}(x_k^\theta-x_j^\theta)(y_k^\theta-y_j^\theta)\prod_{j=1}^N\omega_1(x_j)\omega_2(y_j)dx_jdy_j
\end{align}
with some nonnegative weight functions $\omega_1$ and $\omega_2$. It was also shown in \cite{forrester19} that if the weight functions are specifically chosen as Laguerre weights, i.e. $\omega_1(x)=x^ae^{-x}$ and $\omega_2(y)=y^be^{-y}$, then the hard edge kernel of this model can be depicted in terms of the Fox H-kernel, generalising the original work about the Meijer G-kernel \cite{bertola09}. However, the jPDF in \eqref{cauchy} is not the most general case about the Cauchy two-matrix model. In other words, one can consider a two-parameter generalisation
\begin{align*}
\frac{\prod_{1\leq j<k\leq N}(x_k-x_j)(y_k-y_j)}{\prod_{j,k=1}^N(x_j+y_k)}\prod_{1\leq j<k\leq N}(x_k^\tq-x_j^\tq)(y_k^\tw-y_j^\tw)\prod_{j=1}^N\omega_1(x_j)\omega_2(y_j)dx_jdy_j,
\end{align*}
to distinguish the eigenvalues $\{x_k\}_{k=1}^N$ and $\{y_k\}_{k=1}^N$ not only from the weight functions, but the interactions within themselves.
Obviously, the averaged characteristic polynomials and corresponding Christoffel-Darboux kernel should be related to a two-parameter Cauchy bi-orthogonal polynomials, which motivates us to study the properties of these polynomials and the integrable systems behind this model.

This paper is organised as follows. In Section \ref{pol}, we'd like to make some discussions about the generalised Cauchy bi-orthogonal polynomials, which give us general spectral problems related to \eqref{sspectral}. Moreover, with the Laguerre weight, these generalised Cauchy bi-orthogonal polynomials can be connected with bi-orthogonal Jacobi polynomials. Besides, we show that these generalised bi-orthogonal polynomials can be written as a series sum and furthermore as a contour integral. In Section \ref{is}, we are devoted to the time evolutions and consider how to derive the corresponding integrable systems. Without symmetry property, the derivation becomes rather tough and the method in this paper is totally different with the ones shown in \cite{li19,miki11}. The simplest asymmetric case is considered to illustrate the connection with C-Toda lattice and some concluding remarks are given in the end.

\section{A two-parameter generalisation of Cauchy bi-orthogonal polynomials}\label{pol}
\subsection{Orthogonality and recurrence relation}
Consider the inner product $\langle\cdot,\cdot\rangle$ from $\mathbb{R}[x]\times \mathbb{R}[y]\to\mathbb{R}$ such that
\begin{align}\label{moments}
\langle x^i, y^j\rangle=\int_{\mathbb{R}_+\times \mathbb{R}_+}\frac{x^iy^j}{x+y}\du_1(x)\du_2(y):=m_{i,j}
\end{align}
with two non-negative measures $d\mu_1$ and $d\mu_2$,
then we can define a family of monic two-parameter Cauchy bi-orthogonal polynomials  $\{P_n,\, Q_n\}_{n=0}^\infty$, satisfying the orthogonal relation
\begin{align}\label{ls}
\langle P_n(x^\tq), Q_m(y^\tw)\rangle=h_n\delta_{n,m},\quad \text{for some $h_n$>0}.
\end{align}
Therefore, from the linear system \eqref{ls}, we can get a closed form for these polynomials, showing
\begin{align}\label{ops}
\begin{aligned}
P_n(x^\tq)&=\frac{1}{\tau_n}\left|
\begin{array}{cccc}
m_{0,0}&\cdots&m_{0,(n-1)\tw}&1\\
m_{\tq,0}&\cdots&m_{\tq,(n-1)\tw}&x^\tq\\
\vdots&&\vdots&\vdots\\
m_{n\tq,0}&\cdots&m_{n\tq,(n-1)\tw}&x^{n\tq}
\end{array}
\right|,\\
Q_n(y^\tw)&=\frac{1}{\tau_n}\left|
\begin{array}{cccc}
m_{0,0}&m_{0,\tw}&\cdots&m_{0,n\tw}\\
\vdots&\vdots&&\vdots\\
m_{(n-1)\tq,0}&m_{(n-1)\tq,\tw}&\cdots&m_{(n-1)\tq,n\tw}\\
1&y^\tw&\cdots&y^{n\tw}
\end{array}
\right|,
\end{aligned}
\end{align}
where $\tau_n$ is the normalisation factor $\det(m_{(k-1)\tq,(l-1)\tw})_{k,l=1}^n$. By direct computations, one can show $h_n=\tau_{n+1}/\tau_n$. Moreover, the existence and uniqueness of the polynomials defined by the linear system \eqref{ls} are equal to the  condition $\tau_n\not=0$,
which could be verified by using Andr\'eief formula and shown that
\begin{align*}
\tau_n=\frac{1}{(n!)^2}\int_{\mathbb{R}_+^n\times \mathbb{R}_+^n}\det\left[\frac{1}{x_j+y_k}
\right]_{j,k=1}^n\Delta_n(x^\tq)\Delta_n(y^\tw)\prod_{j=1}^nd\mu_1(x_j)d\mu_2(y_j),
\end{align*}
where $\Delta_n(z)=\prod_{1\leq j<k\leq n}(z_k-z_j)$ for the non-negative measures $d\mu_1$ and $d\mu_2$.

Firstly, we'd like to consider the recurrence relations of these two-parameter polynomials with general weights, which is of essential importance to clarify the properties of these families of polynomials. For simplicity, we'd like to constrain ourselves to the case $\theta_i=1/k_i,\, k_i\in\mathbb{Z}_{>0}$ throughout the paper.
\begin{proposition}\label{rr}
For the two-parameter Cauchy bi-orthogonal polynomials $\{P_n(x^\tq)\}_{n=0}^\infty$, they have the following $(k_1+k_2+2)$-term recurrence relation
\begin{align}\label{rr1}
x\left(
P_{n+1}(x^\tq)+a_nP_n(x^\tq)\right)=\sum_{\al=n-k_2}^{n+k_1+1}\eta_{n,\al}P_\al(x^\tq),
\end{align}
where
\begin{align*}
a_n=-\frac{\int_{\mathbb{R}_+}P_{n+1}(x^\tq)\du_1(x)}{\int_{\mathbb{R}_+}P_{n}(x^\tq)\du_1(x)},\quad \eta_{n,\al}=\frac{\langle x(P_{n+1}(x^\tq)+a_nP_n(x^\tq)),  Q_\al(y^\tw)\rangle}{\langle
P_\al(x^\tq), Q_\al(y^\tw)
\rangle}.
\end{align*}
\end{proposition}
\begin{proof}
The proof is based on the speciality of $a_n$ and the orthogonality of these polynomials. By making the use of $a_n$, one can find
\begin{align*}
\langle x(P_{n+1}(x^\tq)+a_nP_n(x^\tq)),  Q_m(y^\tw)\rangle=-\langle P_{n+1}(x^\tq)+a_nP_n(x^\tq), yQ_m(y^\tw)\rangle.
\end{align*}
Furthermore, since $yQ_m(y^\tw)\in\text{span}\{Q_0(y^\tw),\cdots, Q_{m+k_2}(y^\tw)\}$, we know the above equation equals zero if $m+k_2<n$ according to the orthogonality. Noting that
\begin{align*}
x(P_{n+1}(x^\tq)+a_nP_n(x^\tq))=\sum_{\be=0}^{n+k_1+1}\eta_{n,\be}P_\be(x^\tq) \end{align*}
and \begin{align*} \langle x(P_{n+1}(x^\tq)+a_nP_n(x^\tq)),  Q_m(y^\tw)\rangle=0
\quad \text{if $m<n-k_2$},\end{align*}
we know $\eta_{n,\be}=0$ if $\be<n-k_2$ and the coefficients of $\{\eta_{n,\al}\}_{\al=n-k_2}^{n+k_1+1}$ can be computed from the orthogonality.
\end{proof}

\begin{corollary}
There is a dual recurrence relation for the polynomials $\{Q_n(y^\tw)\}_{n=0}^\infty$
\begin{align}\label{rr2}
y\left(
Q_{n+1}(y^\tw)+\hat{a}_nQ_n(y^\tw)
\right)=\sum_{\al=n-k_1}^{n+k_2+1}\hat{\eta}_{n,\al}Q_\al(y^\tw),
\end{align}
where
\begin{align*}
\hat{a}_n=-\frac{\int_{\mathbb{R}_+}Q_{n+1}(y^\tw)\du_2(y)}{\int_{\mathbb{R}_+}Q_{n}(y^\tw)\du_2(y)},\quad \hat{\eta}_{n,\al}=\frac{\langle P_\al(x^\tq), y(
Q_{n+1}(y^\tw)+\hat{a}_nQ_n(y^\tw))\rangle}{\langle
P_\al(x^\tq), Q_\al(y^\tw)
\rangle}.
\end{align*}
\end{corollary}
Therefore, the recurrence relations \eqref{rr1} and \eqref{rr2} are in fact the spectral problems \eqref{sspectral} if we write the spectral problem in matrix form and introduce the shift operator. Moreover, the constraint in $a_{l+1}a_{l+2}\not=0$ in \eqref{sspectral} is again equal to the constraint the $\tau_n\not=0$, which we have discussed before.

\subsection{Special case: two-parameter Cauchy-Laguerre bi-orthogonal polynomials}
To claim the importance of these polynomials, in this subsection we'd show that with Laguerre weights, these polynomials are related to bi-orthogonal Jacobi polynomials in \cite{borodin98,madhekar82}, which could be used in the further studies about the hard edge scaling limit of the two-parameter Cauchy-Laguerre matrix model. Moreover, the partition partition shown in this subsection can be viewed as a special case of what we would discuss in Section \ref{is}.

Let's consider the Laguerre weights $\du_1(x)=x^ae^{-x}dx$ and $\du_2(y)=y^be^{-y}dy$. In this case, the moments in \eqref{moments} could be written as
\begin{align*}
I_{j,k}:=m_{(j-1)\tq, (k-1)\tw}=\int_{\mathbb{R}_+^2}\frac{x^{a+(j-1)\tq}y^{b+(k-1)\tw}}{x+y}e^{-(x+y)}dxdy.
\end{align*}
These moments are two-parameter generalisations of Cauchy bi-orthogonal polynomials which are slightly different from the one-parameter generalisation considered in \cite{forrester19}---unlike the procedure to connect the one-parameter Cauchy bi-orthogonal polynomials with the standard Jacobi polynomials, we'd like to evaluate the moments and connect them with Jacobi bi-orthogonal polynomials \cite{madhekar82}. Consider an evolution of the moments
\begin{align}\label{tm}
J_{j,k}(s)=\int_{\mathbb{R}_+^2}\frac{x^{a+(j-1)\tq}y^{b+(k-1)\tw}}{x+y}e^{-s(x+y)}dxdy,
\end{align}
and by making the use of transformations on variables $x=\tilde{x}/s$ and $y=\tilde{y}/s$, one can show $$J_{j,k}(s)=s^{-(1+a+b+\tq(j-1)+\tw(k-1))}I_{j,k}.$$
Therefore, $$\frac{d}{ds}J_{j,k}(s)=-s^{-(2+a+b+\tq(j-1)+\tw(k-1))}\left(1+a+b+\tq(j-1)+\tw(k-1)\right)I_{j,k}.$$
On the other hand, the derivative of $s$ can directly lead us to
\begin{align}\label{cancel}
\frac{d}{ds}J_{j,k}(s)&=\int_{\mathbb{R}_+^2} x^{a+\tq(j-1)}y^{b+\tw(k-1)}e^{-s(x+y)}dxdy=
\int_{\mathbb{R}_+}x^{a+\tq(j-1)}e^{-sx}dx
\int_{\mathbb{R}_+}y^{b+\tw(k-1)}e^{-sy}dy\nonumber\\
&=-s^{(2+a+b+\tq(j-1)+\tw(k-1))}\Gamma(1+a+\tq(j-1))\Gamma(1+b+\tw(k-1)).
\end{align}
The equivalence of these two expressions gives rise to
\begin{align*}
I_{j,k}=\frac{\Gamma(1+a+\tq(j-1))\Gamma(1+b+\tw(k-1))}{1+a+b+\tq(j-1)+\tw(k-1)}.
\end{align*}
If we consider two systems of functions \cite[Prop. 3.3]{borodin98}
\begin{align*}
\xi_n(x^\tq)&=\sum_{i=0}^n(-1)^{i}\frac{\left((1+a+b+\tq i)/\tw\right)_n}{i!(n-i)!}x^{\tq i}:=\sum_{i=0}^n c_{n,i}x^{\tq i},\\ \psi_n(x^\tw)&=\sum_{i=0}^n (-1)^{i}\frac{\left((1+a+b+\tw i)/\tq\right)_n}{i!(n-i)!}x^{\tw i}:=\sum_{i=0}^n d_{n,i}x^{\tw i},
\end{align*}
where $(a)_m=a(a+1)\cdots(a+m-1)$ stands for the Pochhammer symbol, then
one can check that $\{\xi_n,\,\psi_n\}_{n=0}^\infty$ are bi-orthogonal in $L^2\left([0,1],x^{a+b}dx\right)$. In other words,
\begin{align*}
\int_0^1 \xi_n(x^\tq)\psi_m(x^\tw)x^{a+b}dx=\tilde{h}_n\delta_{n,m},\quad \tilde{h}_n=\frac{1}{1+a+b+(\tq+\tw)n}.
\end{align*}
The moments of the above bi-orthogonal polynomials are given by
\begin{align*}
\int_0^1 x^{\tq(j-1)}x^{\tw(k-1)}x^{a+b}dx=\frac{1}{1+a+b+\tq(j-1)+\tw(k-1)}.
\end{align*}
Furthermore, if we set
\begin{align}\label{cbop1}
\hat{P}_n(x^\tq)=\sum_{l=0}^n \frac{c_{n,l}}{\Gamma(1+a+\tq l)}x^{\tq l},\quad \hat{Q}_n(y^\tw)=\sum_{l=0}^n \frac{d_{n,l}}{\Gamma(1+b+\tw l)}y^{\tw l},
\end{align}
then it is found that
\begin{align*}
\int_{\mathbb{R}_+^2} &\frac{e^{-(x+y)}}{x+y} \hat{P}_m(x^\tq)\hat{Q}_n(y^\tw) x^ay^b dxdy=\sum_{j=0}^m\sum_{k=0}^n \frac{c_{m,j}d_{n,k}}{\Gamma(1+a+\tq j)\Gamma(1+b+\tw k)}I_{j+1,k+1}\\
&=\sum_{j=0}^m\sum_{k=0}^n c_{m,j}d_{n,k}I_{j+1,k+1}=\int_{0}^1 \xi_m(x^\tq)\psi_n(x^\tw)x^{a+b}dx=\tilde{h}_n\delta_{n,m}.
\end{align*}
Therefore, the polynomials defined in \eqref{cbop1} are Cauchy bi-orthogonal polynomials with Laguerre weight. To make it monic, let's take
\begin{align*}
P_n(x^\tq)&=\frac{\Gamma(1+a+\tq n)}{c_{n,n}}\hat{P}_n(x^\tq)=\frac{\Gamma(n+1)\Gamma(1+a+\tq n)\Gamma((1+a+b+\tq n)/\tw)}{\Gamma\left((1+a+b+(\tq+\tw)n)/\tw\right)}\hat{P}_n(x^\tq),\\
Q_n(y^\tw)&=\frac{\Gamma(1+b+\tw n)}{d_{n,n}}\hat{Q}_n(y^\tw)=\frac{\Gamma(n+1)\Gamma(1+b+\tw n)\Gamma((1+a+b+\tw n)/\tq)}{\Gamma\left((1+a+b+(\tq+\tw)n)/\tq\right)}\hat{Q}_n(y^\tw),
\end{align*}
and thus $h_n$ in orthogonal relation \eqref{ls} is
\begin{align*}
\frac{(\Gamma(n+1))^2\Gamma(1+a+\tq n)\Gamma(1+b+\tw n)\Gamma((1+a+b+\tq n)/\tw)\Gamma((1+a+b+\tw n)/\tq)}{(1+a+b+(\tq+\tw)n)\Gamma((1+a+b+(\tq+\tw)n)/\tw)\Gamma((1+a+b+(\tq+\tw)n)/\tq)}.
\end{align*}
Moreover, the two-parameter Cauchy bi-orthogonal polynomials can be written as contour integral expression
\begin{align*}
\hat{P}_n(x^\tq)=\int_\gamma \frac{du}{2\pi i}\frac{\Gamma((1+a+b-\tq u)/\tw+n)\Gamma(u)}{\Gamma(1+n+u)\Gamma(1+a-\tq u)\Gamma((1+a+b-\tq u)/\tw)}x^{-\tq u},
\end{align*}
where $\gamma$ is a contour encloses $\{0,-1,\cdots,-n\}$ and $\hat{Q}_n(y^\tw)=\hat{P}_n(x^\tq)|_{a\leftrightarrow b, \, x\leftrightarrow y,\, \tq\leftrightarrow\tw}$. These polynomials can be expressed in terms of Fox H-functions as well and in the special symmetric case, the expressions degenerate to \cite[Eqs. (2.8a)-(2.8b)]{forrester19}.

\section{Time evolutions and integrable lattices}\label{is}
In this section, let's consider such time evolutions in the measures that
\begin{align}\label{td}
\begin{aligned}
&I_{j,k}(t)=\langle x^{\tq(j-1)}, y^{\tw(k-1)}\rangle=\int_{\mathbb{R}_+^2}\frac{x^{\tq(j-1)}y^{\tw(k-1)}}{x+y}\du_1(x;t)\du_2(y;t),\\
&\text{$\p_{t}\du_1(x;t)=x\du_1(x;t)$,\quad $\p_t\du_2(y;t)=y\du_2(y;t)$.}
\end{aligned}
\end{align}
This kind of assumption is made to decouple the double integral into single integrals, like the rank $1$ shift condition discussed in \cite{bertola10,li20}, which is useful to compute in terms of $x$ and $y$ seperately (c.f. Equation \eqref{cancel}). Specifically, one of the special solutions of $d\mu_1$ and $d\mu_2$ is $\du_1(x;t)=e^{tx}\du_1(x)$ and $\du_2(y;t)=e^{ty}\du_2(y)$ and if we take $\du_1(x;t)=e^{tx}x^adx$ and $\du_2(y;t)=e^{ty}y^bdy$, then the moments $I_{j,k}(t)$ are the same with $J_{j,k}(-t)$ defined in \eqref{tm}. However, the measures $\du_1(x;t)$ and $\du_2(y;t)$ we consider are arbitrary $L^2$-integrable measures satisfying \eqref{td}, and thus, the partition function of the two-parameter Cauchy-Laguerre matrix model is a special tau function we discuss below.

From the time-dependent moments/inner product \eqref{td}, we can similarly define a family of time-dependent two-parameter Cauchy bi-orthogonal polynomials $\{P_n(x^\tq;t)\}_{n=0}^\infty$ and $\{Q_n(y^\tw;t)\}_{n=0}^\infty$ via the orthogonal relation
\begin{align}\label{ls2}
\text{$\langle P_n(x^{\tq};t), Q_m(y^{\tw};t)\rangle=h_n(t)\delta_{n,m}$ with $h_n=\frac{\tau_{n+1}(t)}{\tau_n(t)}$.}
\end{align}
Please note that $\{P_n(x^\tq;t)\}_{n=0}^\infty$ (resp. $\{Q_n(y^\tw;t)\}_{n=0}^\infty$) also admit $(k_1+k_2+2)$-term recurrence relations following the Proposition \ref{rr} with the coefficients $a_n$, $\hat{a}_n$, $\eta_{n,\al}$ and $\hat{\eta}_{n,\al}$ time-dependent.
To derive the corresponding integrable system, we need the following derivative formula for the two-parameter Cauchy bi-orthogonal polynomials.
\begin{proposition}\label{dert1}
We have the time evolution equation
\begin{align}\label{te}
\p_t\left(
P_{n+1}(x^\tq;t)+a_nP_n(x^\tq;t)
\right)=\frac{\p_t(a_nh_n)}{h_n}P_n(x^\tq;t).
\end{align}
\end{proposition}
\begin{proof}
Since the polynomials are monic and $\{P_k(x^\tq;t)\}_{k=0}^n$ expand a basis of polynomial at order $n$, we have
\begin{align*}
\p_t\left(
P_{n+1}(x^\tq;t)+a_nP_n(x^\tq;t)
\right)=\sum_{k=0}^n \xi_{n,k}(t)P_k(x^\tq;t).
\end{align*}
Moreover, from the orthogonal relation \eqref{ls2}, we know
\begin{align}\label{te2}
\p_{t}\left(\langle P_{n+1}(x^\tq;t)+a_nP_n(x^\tq;t),Q_m(y^\tw;t)\rangle\right)=\p_t h_{n+1}\delta_{n+1,m}+\p_t(a_nh_n)\delta_{n,m}.
\end{align}
The left hand side can be equivalently expressed as
\begin{align*}
\langle \p_t(P_{n+1}(x^\tq;t)+a_nP_n(x^\tq;t)), Q_m(y^\tw;t)\rangle+\langle P_{n+1}(x^\tq;t)+a_nP_n(x^\tq;t),\p_tQ_m(y^\tw;t)\rangle,
\end{align*}
whose second term is zero if $m<n$. Therefore, one can find
\begin{align*}
\langle \p_t(P_{n+1}(x^\tq;t)+a_nP_n(x^\tq;t)), Q_m(y^\tw;t)\rangle=0\quad \text{if $m<n$},
\end{align*}
to conclude that $\xi_{n,m}=0$ if $m<n$. Moreover, if $m=n$, one can find $\xi_{n,n}=\p_t(a_nh_n)/h_n$ directly from \eqref{te2}, and thus complete the proof.
\end{proof}

Similar with \cite[Prop. 3.2]{li19}, this proposition gives us a time evolution equation for the eigenfunction. However, unlike the method demonstrated therein, it is insufficient to derive an integrable system from the equations \eqref{te} and \eqref{te2} only. The coefficients of $x^{\tq n}$ in $P_{n+1}(x^\tq;t)$ could hardly be expressed as derivatives of tau functions and the equation cannot be easily closed. We need a novel method to derive the integrable lattice. The basic idea is to use the recurrence relation \eqref{rr1} and time evolution equation \eqref{te} and make the use of compatibility condition to formulate the integrable equation.

From the equation \eqref{te} and by the use of recurrence relation with $e_n=\p_{t}(a_nh_n)/h_n$, we can write down
\begin{align}\label{eq1}
\begin{aligned}
x\left[
\p_t( P_{n+1}+a_nP_n)+\frac{e_n}{e_{n-1}}a_{n-1}\p_t\left(
P_n+a_{n-1}P_{n-1}
\right)\right]
=e_n\sum_{\al=n-k_2-1}^{n+k_1}\eta_{n-1,\al}P_\al.
\end{aligned}
\end{align}
Moreover, if we denote $f_n:=e_na_{n-1}/e_{n-1}$, then the left hand side of the above equation can also be written as
\begin{align}\label{eq2}
\begin{aligned}
\p_t& [x(P_{n+1}+a_nP_n)]+f_n\p_t[x(P_n+a_{n-1}P_{n-1})]\\
&=\p_t\left(
\sum_{\al=n-k_2}^{n+k_1+1}\eta_{n,\al}P_\al
\right)+f_n\p_t\left(
\sum_{\al=n-k_2-1}^{n+k_1}\eta_{n-1,\al}P_\al
\right)\\
&=\sum_{\al=n-k_2}^{n+k_1}\left(
\p_t\eta_{n,\al}+f_n\p_t \eta_{n-1,\al}\right)P_\al+f_n\p_t\eta_{n-1,n-k_2-1}P_{n-k_2-1}\\
&\quad+\underline{\p_t{P_{n+k_1+1}}+\sum_{\al=n-k_2}^{n+k_1}(\eta_{n,\al}+f_n\eta_{n-1,\al})\p_tP_\al+f_n\eta_{n-1,n-k_2-1}\p_t P_{n-k_2-1}}.
\end{aligned}
\end{align}
To see the compatibility condition of the equations \eqref{eq1} and \eqref{eq2}, we'd like to express the underlined term
as a linear combination of basis $\{P_\al(x^\tq;t)\}_{\al=0}^{n+k_1}$. Notice that there exist indefinite parameters $\xi_{n-k_2},\cdots,\xi_{n+k_1}$ (where we assume $\xi_{n+k_1+1}=1$) such that
\begin{align*}
\sum_{\al=n-k_2}^{n+k_1+1}\xi_\al\p_t(P_\al+a_{\al-1}P_{\al-1})&=\sum_{\al=n-k_2}^{n+k_1+1}(\xi_\al\p_t a_{\al-1})P_{\al-1}\\
&+\p_t P_{n+k_1+1}+\sum_{\al=n-k_2}^{n+k_1}(\xi_\al+\xi_{\al+1}a_\al)\p_t P_\al+\xi_{n-k_2}a_{n-k_2-1}\p_t P_{n-k_2-1}.
\end{align*}
By using equation \eqref{te}, we know $$\sum_{\al=n-k_2}^{n+k_1+1}\xi_\al\p_t(P_\al+a_{\al-1}P_{\al-1})=\sum_{\al=n-k_2}^{n+k_1+1}\xi_\al e_{\al-1}P_{\al-1}.$$
Therefore, a combination of the above two equations leads us to
\begin{align*}
\p_t P_{n+k_1+1}+\sum_{\al=n-k_2}^{n+k_1}(\xi_\al+\xi_{\al+1}a_\al)\p_t P_\al+\xi_{n-k_2}a_{n-k_2-1}\p_t P_{n-k_2-1}
=\sum_{\al=n-k_2}^{n+k_1+1}\xi_\al(e_{\al-1}-\p_t a_{\al-1})P_{\al-1}.
\end{align*}
Moreover, if we assume that $\xi_{n-k_2},\,\cdots,\,\xi_{n+k_1}$ satisfy
\begin{align*}
\xi_\al+\xi_{\al+1}a_\al=\eta_{n,\al}+f_n\eta_{n-1,\al},\quad \text{for $\al=n-k_2,\,\cdots,\,n+k_1$},
\end{align*}
then the underlined term in \eqref{eq2} can be expressed as
\begin{align*}
\sum_{\al=n-k_2}^{n+k_1+1}\xi_\al(e_{\al-1}-\p_t a_{\al-1})P_{\al-1}+(f_n\eta_{n-1,n-k_2-1}-\xi_{n-k_2}a_{n-k_2-1})\p_t P_{n-k_2-1}.
\end{align*}
Hence, we can rewrite the equation \eqref{eq2} as
\begin{align}\label{eq3}
\begin{aligned}
\p_t x(P_{n+1}+a_nP_n)&+f_n\p_tx(P_n+a_{n-1}P_{n-1})\\
&=\sum_{\al=n-k_2}^{n+k_1}\left(\p_t \eta_{n,\al}+f_n\p_t \eta_{n-1,\al}+\xi_{\al+1}(e_\al-\p_t a_\al)\right)P_\al\\
&\quad+\left(f_n\p_t \eta_{n-1,n-k_2-1}+\xi_{n-k_2}(e_{n-k_2-1}-\p_t a_{n-k_2-1})\right)P_{n-k_2-1}\\
&\quad +(f_n\eta_{n-1,n-k_2-1}-\xi_{n-k_2}a_{n-k_2-1})\p_t P_{n-k_2-1},
\end{aligned}
\end{align}
where the last term is a linear combination of $\{P_\al(x^\tq;t)\}_{\al=0}^{n-k_2-2}$. Therefore, according to the independent of the basis, one can finally arrive at the compatibility condition
\begin{align}\label{gct}
\left\{
\begin{aligned}
&e_n\eta_{n-1,\al}=\p_t \eta_{n,\al}+f_n\p_t \eta_{n-1,\al}+\xi_{\al+1}(e_\al-\p_t a_\al),\quad \al=n-k_2,\cdots,n+k_1,\\
&e_n\eta_{n-1,n-k_2-1}=f_n\p_t\eta_{n-1,n-k_2-1}+\xi_{n-k_2}(e_{n-k_2-1}-\p_t a_{n-k_2-1}),\\
&0=f_n\eta_{n-1,n-k_2-1}-\xi_{n-k_2}a_{n-k_2-1},
\end{aligned}
\right.
\end{align}
where $\{\xi_{\al},\,\al=n-k_2,\,\cdots,\,n+k_1\}$ satisfy the linear system $\xi_\al+\xi_{\al+1}a_\al=\eta_{n,\al}+f_n\eta_{n-1,\al}$ with $\xi_{n+k_1+1}=1$.
The equations for $P_n(x^\tq;t)$ are not closed in this case. One should consider dual equations for $Q_n(y^\tw;t)$. By making the use of the spectral problem \eqref{rr2} and time evolution
\begin{align*}
\partial_t\left(
Q_{n+1}(y^{\tw};t)+\hat{a}_nQ_n(y^\tw;t)
\right)=\frac{\partial_t(\hat{a}_nh_n)}{h_n}:=\hat{e}_nQ_n(y^\tw;t),
\end{align*}
one can get a dual equation for \eqref{gct} by changing $a\mapsto \hat{a}$, $e\mapsto \hat{e}$, $f\mapsto \hat{f}$, $\xi\mapsto \hat{\xi}$ and $\eta\mapsto \hat{\eta}$. After combining the equations held by ${P_n(x^\tq;t)}_{n=0}^\infty$ and $\{Q_n(y^\tw;t)\}_{n=0}^\infty$, the lattice equations are closed and an example is illustrated in the following subsection. 

\begin{remark} 
Even in the case $\tq=\tw=1$, if $d\mu_1\not= d\mu_2$ (i.e. the moments are not symmetric), the time-dependent partition function of the original Cauchy two-matrix model doesn't satisfy the CKP hierarchy\footnote{One can follow the procedure exhibited in \cite{wang10} and take the Gram determinant into the first equation of CKP equation (i.e. \cite[Eq. (3)]{wang10}). It is not difficult to see the asymmetric Gram determinant doesn't satisfy the bilinear equation of CKP equation.}.
However, it is interesting to see, by the use of the Andrei\'ef formula, the partition function can be written as
\begin{align*}
\tau_n&=\frac{1}{(n!)^2}\int_{\mathbb{R}_+^n\times \mathbb{R}_+^n}\det\left[\frac{1}{x_j+y_k}
\right]_{j,k=1}^n\Delta_n(x)\Delta_n(y)\prod_{j=1}^n d\mu_1(x_j)d\mu_2(y_j)\\&=\det\left[
\int_{\mathbb{R}_+^2}\frac{x^{j-1}y^{k-1}}{x+y}\du_1(x;t)\du_2(y;t)
\right]_{j,k=1}^n:=\det\left[
I_{j,k}
\right]_{j,k=1}^n,
\end{align*}
and the moments $I_{j,k}$ are defined by
\begin{align*}
I_{j,k}=\int_{-\infty}^t \phi_j(t)\psi_k(t)dt, \quad \text{where $\phi_j(t)=\int_{\mathbb{R}_+}x^{j-1}d\mu_1(x;t)$ and $\psi_k(t)=\int_{\mathbb{R}_+}y^{k-1}d\mu_2(y;t)$.}
\end{align*}
The determinant with this kind of moment is usually called as the Gram determinant in the soliton theory; please see \cite[\S 2]{hirota04} for details. Moreover, if the time-dependent measures are taken as $d\mu_i(x;t)=\sum_{k=1}^\infty t_{2k+1}x^{2k+1}d\mu_i(x)$ with $i=1,\,2$, then the Gram determinant satisfies the KP hierarchy with odd flows, as shown in \cite[\S 3.2]{hirota04}. Therefore, we can tell that even though the asymmetric moments are not related to the CKP hierarchy any more, there is another integrable hierarchy behind this random matrix model, which is a special case of KP hierarchy with odd flows only. Based on these facts, we would like to demonstrate the first non-trivial asymmetric case and show its exact solvability.
\end{remark}

\subsection{The first non-trivial asymmetric tau function and integrable lattice}
In this part, we'd like to consider the first nontrivial asymmetric tau function, i.e. $\theta_1=\theta_2=1$ case with $d\mu_1\not=d\mu_2$. In this setting, we can define the moments
\begin{align*}
m_{i,j}=\langle x^i, y^j\rangle=\int_{\mathbb{R}_+\times\mathbb{R}_+}\frac{x^iy^j}{x+y}d\mu_1(x;t)d\mu_2(y;t)
\end{align*}
and the corresponding $\tau$-function 
\begin{align}\label{tau1}
\tau_n=\det(m_{j,k})_{j,k=0}^{n-1}.
\end{align}
Similar to \eqref{ops}, one can define the asymmetric Cauchy bi-orthogonal polynomials $\{P_n(x;t)\}_{n=0}^\infty$ and $\{Q_n(y;t)\}_{n=0}^\infty$ such that 
\begin{align*}
\text{
$\langle P_n(x;t), Q_m(y;t)\rangle=h_n\delta_{n,m}$ with $h_n=\tau_{n+1}/\tau_n$}.
\end{align*}
Furthermore, these polynomials admit the following properties.
\begin{proposition}
 $P_n(x;t)$ and $Q_n(y;t)$ satisfy the recurrence relations
\begin{align*}
x(P_{n+1}(x;t)+a_nP_n(x;t))&=P_{n+2}(x;t)+b_nP_{n+1}(x;t)+c_nP_{n}(x;t)+d_nP_{n-1}(x;t),\\
y(Q_{n+1}(y;t)+\hat{a}_nQ_n(y;t))&=Q_{n+2}(y;t)+\hat{b}_nQ_{n+1}(y;t)+\hat{c}_nQ_n(y;t)+\hat{d}_nQ_{n-1}(y;t),
\end{align*}
where the coefficients are given by
\begin{align}\label{changeofvariable}
\begin{aligned}
&a_n=\hat C_n,\, b_n=\hat C_n+B_{n+1},\, c_n=-A_{n+1}-\hat C_n\hat B_n,\, d_n=-A_n\hat C_n,\\
&\hat{a}_n= C_n,\,\hat{b}_n={C}_n+\hat B_{n+1},\,\hat{c}_n=-A_{n+1}- C_n B_n,\,\hat{d}_n=-A_n C_n,
\end{aligned}
\end{align}
with 
\begin{align}\label{change}
A_n=\frac{\tau_{n+1}\tau_{n-1}}{\tau_n^2},\, B_n=\frac{\xi_{n+1}}{\tau_{n+1}}-\frac{\xi_n}{\tau_n},\, \hat B_n=\frac{\hat{\xi}_{n+1}}{\tau_{n+1}}-\frac{\hat{\xi}_n}{
\tau_n},\, C_n=-\frac{{{{\sigma}}}_{n+1}\tau_n}{{\sigma}_n\tau_{n+1}},\, \hat C_n=-\frac{{\hat{\sigma}}_{n+1}\tau_n}{{\hat{\sigma}}_n\tau_{n+1}}
\end{align}
and the functions $\sigma$ and $\xi$ (resp. $\hat{\sigma}$ and $\hat{\xi}$) have the determinant expressions
\begin{align}\label{taufunctions}
\begin{aligned}
\sigma_n&=\det\left(\begin{array}{c}
m_{i,j}\\
\phi_j\end{array}
\right)_{\substack{i=0,\cdots,n-1\\j=0,\cdots,n}},\quad
\xi_n=\det\left(
m_{i,j}
\right)_{\substack{i=0,\cdots,n-2,n\\j=0,\cdots,n-1}},
\\
 \hat{\sigma}_n&=\det\left(
m_{i,j}\, \hat{\phi}_i
\right)_{\substack{i=0,\cdots,n \\j=0,\cdots,n-1}},\quad \hat{\xi}_n=\det\left(
m_{i,j}
\right)_{\substack{i=0,\cdots,n-1\\j=0,\cdots,n-2,n}}.
\end{aligned}
\end{align}
The single moments $\phi_j$ and $\hat{\phi}_j$ are defined by $\hat{\phi}_j=\int_{\mathbb{R}_+}x^jd\mu_1(x;t)$ and ${\phi}_j=\int_{\mathbb{R}_+}y^jd\mu_2(y;t)$.
\end{proposition}
This proposition is a direct consequence of Prop. \ref{rr} and we omit its proof here. Moreover, it follows from the derivative formula in Prop. \ref{dert1} and we have the following proposition.
\begin{proposition}\label{prop_evo}  $P_n(x;t)$ and $Q_n(y;t)$ evolve as
\begin{align*}
\partial_t 
P_{n+1}(x;t)+\hat C_n\partial_t P_n(x;t)&=\hat C_n(B_n+\hat B_n)P_n(x;t),\\ \partial_t
Q_{n+1}(y;t)+{C}_n\partial_t Q_n(y;t)&= C_n(B_n+\hat B_n)Q_n(y;t).
\end{align*}
\end{proposition}
\begin{proof} Compared with Prop. \ref{dert1}, we are left to prove that 
$$\partial_t \log h_n=B_n+\hat B_n,$$
which is equivalent to 
$$\partial_t \tau_n=\xi_n+\hat \xi_n.$$ This can be immediately achieved by following Lemma 3.2 in \cite{chang18}.
\end{proof}
Following the procedure demonstrated above, one can derive the following nonlinear lattice system from the compatibility condition
\begin{subequations}
\begin{align}
&\partial_t A_n=A_n(B_n+\hat B_n-B_{n-1}-\hat B_{n-1})\label{eq_non},\\
&\partial_t B_n=(B_{n-1}+\hat B_{n-1})\hat C_{n-1}-(B_{n}+\hat B_{n})\hat C_n,\\
&\partial_t \hat B_n=(B_{n-1}+\hat B_{n-1})C_{n-1}-(B_{n}+\hat B_{n}) C_n,\\
&\partial_t C_n=C_n\left(C_n-\frac{A_n}{C_{n-1}}-\hat B_{n}-C_{n+1}+\frac{A_{n+1}}{C_n}+\hat B_{n+1}\right),\\
&\partial_t \hat C_n=\hat C_n\left(\hat C_n-\frac{A_n}{\hat C_{n-1}}- B_{n}-\hat C_{n+1}+\frac{A_{n+1}}{\hat C_n}+ B_{n+1}\right).\label{eq_non2}
\end{align}
\end{subequations}
We need to emphasise that although the recurrence relations \eqref{changeofvariable} seem to admit eight coefficients, in fact, they do have only five independent coefficients expressed by $A_n$, $B_n$, $\hat{B}_n$, $C_n$ and $\hat{C}_n$. It means that the equations held by these recurrence coefficients should be only five equations rather than eight equations; the other three are automatically compatible. To summarise, we have the following proposition.
\begin{theorem}\label{prop2}
The system \eqref{eq_non}-\eqref{eq_non2} admit the matrix integrals solution
\begin{align*}
&\tau_n=\frac{1}{(n!)^2}\int_{\mathbb{R}^n_+\times\mathbb{R}^n_+}\frac{\Delta^2_n(x)\Delta^2_n(y)}{\prod_{j,k=1}^n(x_j+y_k)}\prod_{j=1}^nd\mu_1(x_j;t)d\mu_2(y_j;t);\\
&\hat{\sigma}_n=\frac{1}{n!(n+1)!}\int_{\mathbb{R}_+^{n+1}\times\mathbb{R}_+^n}\frac{\Delta^2_{n+1}(x)\Delta^2_n(y)}{\prod_{j=1}^{
n+1}\prod_{k=1}^n(x_j+y_k)}\prod_{j=1}^{n+1}\prod_{k=1}^nd\mu_1(x_j;t)d\mu_2(y_k;t);\\
&{\sigma}_n=\frac{1}{n!(n+1)!}\int_{\mathbb{R}_+^{n}\times\mathbb{R}_+^{n+1}}\frac{\Delta^2_{n}(x)\Delta^2_{n+1}(y)}{\prod_{j=1}^{
n}\prod_{k=1}^{n+1}(x_j+y_k)}\prod_{j=1}^{n}\prod_{k=1}^{n+1}d\mu_1(x_j;t)d\mu_2(y_k;t);\\
&\hat{\xi}_n=\frac{1}{(n!)^2}\int_{\mathbb{R}_+^n\times\mathbb{R}_+^n}\frac{\Delta_n^2(x)\Delta_n^2(y)}{\prod_{j,k=1}^n(x_j+y_k)}\sum_{j=1}^ny_j\prod_{j=1}^nd\mu_1(x_j;t)d\mu_2(y_j;t);\\
&{\xi}_n=\frac{1}{(n!)^2}\int_{\mathbb{R}_+^n\times\mathbb{R}_+^n}\frac{\Delta_n^2(x)\Delta_n^2(y)}{\prod_{j,k=1}^n(x_j+y_k)}\sum_{j=1}^nx_j\prod_{j=1}^nd\mu_1(x_j;t)d\mu_2(y_j;t)
\end{align*}
with the variable transformation \eqref{change} and the measures satisfy \eqref{td}.\end{theorem}
 \begin{proof} 
 Firstly noted that the tau functions admit the determinant formulae \eqref{tau1} and \eqref{taufunctions}. The rest part of the theorem can be proved by using the Andr\'eief formula \cite{andreief86,forrester18}.
Here we only give a brief proof of $\xi_n$ ($\hat{\xi}_n$ can be proved in a similar way) and the derivations of $\tau_n$ and $\sigma_n$ (resp. $\hat{\sigma}_n$) can be found in \cite[Eq. (2.5), (3.9)]{bertola10}.

 Since 
 \begin{align*}
 \xi_n&=\det\left(
m_{i,j}
\right)_{\substack{i=0,\cdots,n-2,n\\j=0,\cdots,n-1}}=\det\left(
\int_{\mathbb{R}_+\times\mathbb{R}_+}\frac{x^iy^j}{x+y}d\mu_1(x;t)d\mu_2(y;t)
\right)_{\substack{i=0,\cdots,n-2,n\\j=0,\cdots,n-1}}\\
&=\frac{1}{(n!)^2}\int_{\mathbb{R}_+^n\times\mathbb{R}_+^n}\det\left(
\phi_i(x_j)
\right)_{\substack{i=0,\cdots,n-2,n\\j=0,\cdots,n-1}}\det\left(
\phi_i(y_j)
\right)_{i,j=0}^{n-1}\det\left(
\frac{1}{x_i+y_j}
\right)_{i,j=0}^{n-1}\prod_{i,j=0}^{n-1}d\mu_1(x_i;t)d\mu_2(y_j;t)
\end{align*}
with $\phi_i(x)=x^i$, the remaining part is to compute the Vandermonde-type determinants. By showing that
\begin{align*}
\det\left(
\phi_i(x_j)
\right)_{\substack{i=0,\cdots,n-2,n\\j=0,\cdots,n-1}}=(\sum_{j=1}^n x_j)\Delta_n(x),
\end{align*}
the proof is completed.
\end{proof}

Despite the nonlinear form, we would like to derive its bilinear form with the help of dependent variable transformations \eqref{change}.
\begin{proposition}
The tau-functions satisfy the following bilinear form 
\begin{subequations}
\begin{align}
&D_t\tau_{n+1}\cdot\tau_{n}=\sigma_n\hat{\sigma}_n,\label{lattice1}\\
&D_t\xi_n\cdot\tau_n=\hat{\sigma}_n\sigma_{n-1},\label{lattice2}\\ 
&D_t\hat{\xi}_n\cdot\tau_n=\sigma_n\hat{\sigma}_{n-1},\label{lattice3}\\
&D_t\xi_{n+1}\cdot \tau_n+D_t\tau_{n+1}\cdot \hat{\xi}_n=\sigma_n\partial_t\hat{\sigma}_n,\label{lattice4}
\\&D_t\hat{\xi}_{n+1}\cdot\tau_n+D_t\tau_{n+1}\cdot\xi_n=\hat{\sigma}_n\partial_t\sigma_n,\label{lattice5}
\end{align}
\end{subequations}
\end{proposition}
\begin{proof}First, by following Lemma 3.2 in \cite{chang18}, we have the derivative formulae as follows:
\begin{align*}
&\partial_t\tau_n=\xi_n+\hat \xi_n=-\tilde\tau_n,\quad \partial_t\xi_n=-\tilde\xi_n,\quad \partial_t\hat\xi_n=-\tilde{\hat\xi}_n,\quad\partial_t\sigma_n=\alpha_n+\beta_n,\quad \partial_t{\hat\sigma}_n=\hat\alpha_n+\hat\beta_n,
\end{align*}
where
\begin{align*}
 &\alpha_n=\det\left(\begin{array}{c}
m_{i,j}\\
\phi_j\end{array}
\right)_{\substack{i=0,\cdots,n-1\\j=0,\cdots,n-1,n+1}},\quad
 \hat{\alpha}_n=\det\left(
m_{i,j}\, \hat{\phi}_i
\right)_{\substack{i=0,\cdots,n-1,n+1 \\j=0,\cdots,n-1}},\\
&\beta_n=\det\left(\begin{array}{c}
m_{i,j}\\
\phi_j\end{array}
\right)_{\substack{i=0,\cdots,n-2,n\\j=0,\cdots,n}},\quad
 \quad \hat{\beta}_n=\det\left(
m_{i,j}\, \hat{\phi}_i
\right)_{\substack{i=0,\cdots,n \\j=0,\cdots,n-2,n}},\\
&\tilde \tau_n=\det\left(\begin{array}{cc}
m_{i,j}&\hat\phi_i\\
\phi_j&0\end{array}
\right)_{\substack{i=0,\cdots,n-1\\j=0,\cdots,n-1}}, \quad
\tilde\xi_n=\det\left(\begin{array}{cc}
m_{i,j}&\hat\phi_i\\
\phi_j&0\end{array}
\right)_{\substack{i=0,\cdots,n-2,n\\j=0,\cdots,n-1}},\\
&\tilde{\hat\xi}_n=\det\left(\begin{array}{cc}
m_{i,j}&\hat\phi_i\\
\phi_j&0\end{array}
\right)_{\substack{i=0,\cdots,n-1\\j=0,\cdots,n-2,n}}.
\end{align*}
The remaining part can be completed by use of the Jacobi identity \cite{hirota04}, which reads
\begin{align}
\mathcal{D} \mathcal{D}\left(\begin{array}{cc}
i_1 & i_2 \\
j_1 & j_2 \end{array}\right)=\mathcal{D}\left(\begin{array}{c}
i_1  \\
j_1 \end{array}\right)\mathcal{D}\left(\begin{array}{c}
i_2  \\
j_2 \end{array}\right)-\mathcal{D}\left(\begin{array}{c}
i_1  \\
j_2 \end{array}\right)\mathcal{D}\left(\begin{array}{c}
i_2  \\
j_1 \end{array}\right). \label{id:jacobi}
\end{align}
Here $\mathcal{D}$ is an indeterminate determinant. $\mathcal{D}\left(\begin{array}{cccc}
i_1&i_2 &\cdots& i_k\\
j_1&j_2 &\cdots& j_k
\end{array}\right)$ with $ i_1<i_2<\cdots<i_k,\ j_1<j_2<\cdots<j_k$ denotes the
determinant of the matrix obtained from $\mathcal{D}$ by removing the rows with indices
$i_1,i_2,\dots, i_k$ and the columns with indices $j_1,j_2,\dots, j_k$.

The first three relations \eqref{lattice1}-\eqref{lattice3} can be derived by taking
\begin{align*}
&\mathcal{D}_1=\det\left(\begin{array}{cc}
m_{i,j}&\hat\phi_i\\
\phi_j&0\end{array}
\right)_{\substack{i=0,\cdots,n\\j=0,\cdots,n}},\quad i_1=j_1=n+1,\,i_2=j_2=n+2,\\
&\mathcal{D}_2=\det\left(\begin{array}{ccc}
m_{i,j}&\hat\phi_i&0\\
\phi_j&0&1\end{array}
\right)_{\substack{i=0,\cdots,n\\j=0,\cdots,n-1}},\quad i_1=n,\, i_2=n+1,\,j_1=n+1,\, j_2=n+2,\\
&\mathcal{D}_3=\det\left(\begin{array}{cc}
m_{i,j}&\hat\phi_i\\
\phi_j&0\\
0&1\end{array}
\right)_{\substack{i=0,\cdots,n-1\\j=0,\cdots,n}},\quad i_1=n+1,\, i_2=n+2,\,j_1=n,\, j_2=n+1,
\end{align*}
respectively.

The validity of \eqref{lattice4} can be confirmed by combining two relations
\begin{align*}
&\tau_n\partial_t \xi_{n+1}=\xi_{n+1}\partial_t\tau_n+\sigma_n\hat\alpha_n,\\
&\hat\xi_n\partial_t \tau_{n+1}=\tau_{n+1}\partial_t\hat\xi_n+\sigma_n\hat\beta_n,
\end{align*}
which are consequences of applying the Jacobi identity to
\begin{align*}
&\mathcal{D}_{4,1}=\det\left(\begin{array}{ccc}
m_{i,j}&\hat\phi_i\\
\phi_j&0\end{array}
\right)_{\substack{i=0,\cdots,n-1,n+1\\j=0,\cdots,n}},\quad i_1= j_1=n+1,\,i_2=j_2=n+2,\\
&\mathcal{D}_{4,2}=\det\left(\begin{array}{cc}
m_{i,j}&\hat\phi_i\\
\phi_j&0\end{array}
\right)_{\substack{i=0,\cdots,n\\j=0,\cdots,n}},\quad i_1=n+1,\, i_2=n+2,\,j_1=n,\, j_2=n+2,
\end{align*}
 respectively.
 Similarly, by combing two relations 
 \begin{align*}
&\tau_n\partial_t \hat\xi_{n+1}=\hat\xi_{n+1}\partial_t\tau_n+\hat\sigma_n\alpha_n,\\
&\xi_n\partial_t \tau_{n+1}=\tau_{n+1}\partial_t\xi_n+\hat\sigma_n\beta_n,
\end{align*}
 obtained by applying the Jacobi identity to
 \begin{align*}
&\mathcal{D}_{5,1}=\det\left(\begin{array}{ccc}
m_{i,j}&\hat\phi_i\\
\phi_j&0\end{array}
\right)_{\substack{i=0,\cdots,n\\j=0,\cdots,n-1,n+1}},\quad i_1= j_1=n+1,\,i_2=j_2=n+2,\\
&\mathcal{D}_{5,2}=\det\left(\begin{array}{cc}
m_{i,j}&\hat\phi_i\\
\phi_j&0\end{array}
\right)_{\substack{i=0,\cdots,n\\j=0,\cdots,n}},\quad i_1=n,\, i_2=n+2,\,j_1=n+1,\, j_2=n+2,
\end{align*}
 respectively, one can prove \eqref{lattice5}.
\end{proof}

The bilinear form \eqref{lattice1}-\eqref{lattice5} has many intriguing properties. Firstly, when we consider the moments are symmetric, i.e. $\phi_i=\hat{\phi}_i$ and $m_{i,j}=m_{j,i}$, then this lattice equation can be degenerated to the C-Toda lattice, which was introduced in \cite{chang18,li19}. The reason is that in the symmetric case, $\xi_n=\hat\xi_n,\sigma_n=\hat\sigma_n$ and $2\xi_n$ can be regarded as the derivative of $\tau_n$. Therefore, equations \eqref{lattice4} and \eqref{lattice5} are the derivatives of \eqref{lattice2} and \eqref{lattice3}, which are naturally valid. At the same time, the nonlinear system can be degenerated as well. The nonlinear system \eqref{eq_non} will reduce to the nonlinear C-Toda lattice in \cite[eq (3.17)]{chang18} by setting $B_n=\hat B_n$, $C_n=\hat C_n=\sqrt{B_{n+1}A_{n+1}/{B_n}}$ in the symmetric case.

 Secondly, this bilinear form can be iterated. Starting with the initial values $\tau_0=\xi_0=\hat{\xi}_0=1$, $\sigma_0=\phi_0$ and $\hat{\sigma}_0=\hat{\phi}_0$, one can get $\tau_1$ from the equation \eqref{lattice1}, $\xi_1$, $\hat{\xi}_1$ from equations \eqref{lattice4} and \eqref{lattice5} and $\sigma_1$, $\hat{\sigma}_1$ from equations \eqref{lattice2} and \eqref{lattice3}. By iterating these equations repeatedly, one can get the general expressions for these variables (i.e. $\tau_n=\det(m_{i,j})_{i,j=0}^{n-1}$ and $\sigma_n$, $\hat{\sigma}_n$, $\xi_n$ and $\hat{\xi}_n$ have expressions in \eqref{taufunctions}), which means that the lattice equation is recursively solvable.

\section{Concluding remarks}
In this article, we focus on the two-parameter generalisation of the Cauchy two-matrix model, its average characteristic polynomials and corresponding integrable lattice. We show that the asymmetric Gram determinant plays an important role in the tau function theory of the integrable lattice \eqref{lattice1}-\eqref{lattice5} and therefore, in the general integrable lattice \eqref{gct} as well as its dual lattice equation. However, it is not the end of the story. In the work \cite{bertola092,forrester19}, it has been shown that with $d\mu_1(x)=x^ae^{-x}dx$ and $d\mu_2(y)=y^{a+1}e^{-y}dy$, the ($\theta$-deformed) Cauchy two-matrix model is related to the ($\theta$-deformed) Bures ensemble, and the partition function of the latter can be regarded as the $\tau$-function of the B-Toda lattice \cite{chang182,li20}. It implies that the asymmetric lattice equation \eqref{lattice1}-\eqref{lattice5} can be degenerated to the B-Toda lattice as well, which has a four-term recurrence relation of the form \eqref{1} with proper coefficients (c.f. \cite[Equation 3.43]{chang182}). Therefore, we'd like to know the mechanism of the symmetric/skew-symmetric reduction and whether there is any $\theta$-deformed Toda lattice of BKP type. The mechanism of reduction would reveal some subgroups of the affine Weyl group $\tilde{W}^{(k,k+1)}(A_l)$, and it is interesting to consider how to make reductions on the discrete eigenfunctions. Regarding the $\theta$-deformed B-Toda lattice, since the tau function of B-type is of Pfaffian form, one needs to consider more Pfaffian techniques to involve in the tau function expression. Besides, as there is an equation connected Sawada-Kotera equation and Kaup-Kuperschmidt equation, we'd like to explore the discrete version to connect B-Toda and C-Toda lattices.

As we showed, the tau-function of the lattice \eqref{lattice1}-\eqref{lattice5} is related to the partition function of the generalised Cauchy two-matrix model. The Cauchy two-matrix model, and corresponding Meijer G-function, recently has been studied in the theories of Hurwitz number and topological recursion \cite{bertola19}. It's still open for us to know whether the rationally weighted Hurwitz number is related to this generalised $\theta$-deformed integrable hierarchy.

\section*{Acknowledgement}
X. Chang is partially supported by National Natural Science Foundation of China (Grant nos.11688101, 11731014 and 11701550) and the Youth Innovation Promotion Association CAS. S. Li  is supported by the ARC Centre of Excellence for Mathematical and Statistical frontiers (ACEMS). S. Tsujimoto is partially supported by JSPS KAKENHI (Grant Nos. 19H01792, 17K18725) and G. Yu is partially supported by National Natural Science Foundation of China (Grant no. 11871336). The authors would also like to thank Profs. Xingbiao Hu and Dafeng Zuo for helpful discussions.

\small
\bibliographystyle{abbrv}

\def\cydot{\leavevmode\raise.4ex\hbox{.}}
  \def\cydot{\leavevmode\raise.4ex\hbox{.}} \def\cprime{$'$}

\end{document}